\documentclass[11pt,a4paper,twoside]{article}
\usepackage{url}
\usepackage{latexsym}
\usepackage{amscd}
\usepackage{amssymb}
\usepackage{amsmath}
\usepackage{color}
\usepackage{float}
\usepackage{longtable}
\usepackage{xspace}
\usepackage{stmaryrd}
\usepackage{color}
\usepackage{graphics}
\usepackage{graphicx}
\usepackage[pdftitle={},pdfauthor={},pdfpagelabels=true,linktocpage=true]{hyperref}
\usepackage{xifthen}
\usepackage{enumitem}
\usepackage[normalem]{ulem}

\newcommand{\F}{\mathbb{F}}
\newcommand{\samples}{\mathop{\stackrel{\ _\$}{\gets}}}

\DeclareMathOperator{\cir}{circ}

\DeclareMathOperator{\wt}{wt}
\newtheorem{thm}{Theorem}[section]
\newtheorem{lem}[thm]{Lemma}
\newtheorem{cor}[thm]{Corollary}
\newtheorem{propo}[thm]{Proposition}
\newtheorem{clm}[thm]{Claim}
\newtheorem{defn}[thm]{Definition}
\newtheorem{assm}[thm]{Assumption}
\newtheorem{rem}[thm]{Remark}
\newtheorem{obs}[thm]{Observation}
\newtheorem{egs}[thm]{Example}

\newtheorem{fct}[thm]{Fact}
\newtheorem{cons}[thm]{Construction}
\newtheorem{nte}[thm]{Note}

\newtheorem{property}[thm]{Property}

\newenvironment{theorem}{\begin{thm}}{\end{thm}}

\newenvironment{definition}{\begin{defn}}{\end{defn}}
\newenvironment{assumption}{\begin{assm}\begin{em}}{\end{em}\end{assm}}

\newenvironment{observation}{\begin{obs}\begin{em}}{\end{em}\end{obs}}

\newlength{\saveparindent}
\newlength{\saveparskip}
\newcount\proofqeded
\newcount\proofended
\def\qed{{\hspace{1pt}\rule[-1pt]{3pt}{9pt}}
\addtolength{\parskip}{-0pt}
\setlength{\parindent}{\saveparindent}
\global\advance\proofqeded by 1 }
\def\qedenv{
\addtolength{\parskip}{-0pt}
\setlength{\parindent}{\saveparindent}
\global\advance\proofqeded by 1 }
\newenvironment{proof}%
 {\proofstart}%
 {\ifnum\proofqeded=\proofended~\qed\fi \global\advance\proofended by 1
  \medskip}
 {\proofenvstart}%
 {\ifnum\proofqeded=\proofended\qedenv\fi \global\advance\proofended by 1
  \medskip}
\makeatletter
\def\proofstart{\@ifnextchar[{\@oprf}{\@nprf}}
\def\proofenvstart{\@ifnextchar[{\@osprf}{\@nsprf}}
\def\@oprf[#1]{\protect\vspace{6pt}\noindent{\bf Proof of #1:\ }%
\addtolength{\parskip}{5pt}\setlength{\parindent}{0pt}}
\def\@osprf[#1]{\protect\vspace{6pt}\noindent
\addtolength{\parskip}{5pt}\setlength{\parindent}{0pt}}
\def\@nprf{\protect\vspace{6pt}\noindent{\bf Proof:\ }%
\addtolength{\parskip}{5pt}\setlength{\parindent}{0pt}}
\def\@nsprf{\protect\vspace{6pt}\noindent%
\addtolength{\parskip}{5pt}\setlength{\parindent}{0pt}}

\newcommand{\calD}{{\cal D}}

\newcommand{\calX}{{\cal X}}
\newcommand{\calY}{{\cal Y}}

\newcommand{\G}{{{\mathbb G}}}
\newcommand{\N}{{{\mathbb N}}}
\newcommand{\Z}{{{\mathbb Z}}}

\newcommand{\Ft}{\mathbb{F}_2} 
\def\bits{\{0,1\}}

\newcommand{\xor}{{\;\oplus\;}}

\def\getsr{\stackrel{{\scriptscriptstyle\$}}{\leftarrow}}

\newcommand{\secpar}{\kappa}
\newcommand{\alice}{\ensuremath{\mathsf{Alice}}\xspace}
\newcommand{\bob}{\ensuremath{\mathsf{Bob}}\xspace}

\newcommand{\env}{\ensuremath{\mathcal{Z}}\xspace}
\newcommand{\adv}{\ensuremath{\mathcal{A}}\xspace}

\newenvironment{boxfig}[2]{
   \begin{figure}[ht!]
     \newcommand{\FigCaption}{#1}
     \newcommand{\FigLabel}{#2}
     \begin{center}
       \begin{small}
       \setlist{nosep}
         \begin{tabular}{@{}|@{~~}l@{~~}|@{}}
           \hline
           \rule[-1.5ex]{0pt}{1ex}\begin{minipage}[b]{0.96\linewidth}
             \smallskip
             }{%
           \end{minipage}\\
           \hline
         \end{tabular}
       \end{small}
       \vspace{-0.3cm}
       \caption{\FigCaption}
       \label{\FigLabel}
     \end{center}
   \end{figure}
}

\DeclareMathOperator{\HW}{\mathsf{HW}}
\newcommand{\st}{\ensuremath{\mathsf{st}}\xspace}
\newcommand{\pk}{\ensuremath{\mathsf{pk}}\xspace}
\newcommand{\sk}{\ensuremath{\mathsf{sk}}\xspace}
\newcommand{\mes}{\ensuremath{\mathsf{m}}\xspace}
\newcommand{\ciph}{\ensuremath{\mathsf{ct}}\xspace}
\newcommand{\pksp}{\ensuremath{\mathcal{PK}}\xspace}
\newcommand{\sksp}{\ensuremath{\mathcal{SK}}\xspace}
\newcommand{\messp}{\ensuremath{\mathcal{M}}\xspace}
\newcommand{\rndsp}{\ensuremath{\mathcal{R}}\xspace}
\newcommand{\ciphsp}{\ensuremath{\mathcal{C}}\xspace}
\newcommand{\kg}{\ensuremath{\mathsf{KG}}\xspace}
\newcommand{\enc}{\ensuremath{\mathsf{Enc}}\xspace}
\newcommand{\dec}{\ensuremath{\mathsf{Dec}}\xspace}
\newcommand{\pke}{\ensuremath{\mathsf{PKE}}\xspace}

\newcommand{\rnd}{\ensuremath{\mathsf{r}}\xspace}

\newcommand{\Fot}{\ensuremath{\mathcal{F}_{\mathrm{OT}}}\xspace}
\newcommand{\snd}{\alice}
\newcommand{\rec}{\bob}
\newcommand{\suc}{\ensuremath{\mathcal{S}}}
\newcommand{\lenf}{\lambda}
\newcommand{\Fro}{\ensuremath{\mathcal{F}_{\mathrm{RO}}}\xspace}
\newcommand{\Frok}{\ensuremath{\mathcal{F}_{\mathrm{RO}1}}\xspace}
\newcommand{\Frop}{\ensuremath{\mathcal{F}_{\mathrm{RO}2}}\xspace}
\newcommand{\Fuc}{\ensuremath{\mathcal{F}\xspace}}

\newcommand{\msg}[2]{%
\ifthenelse{\isempty{#1}}{\ensuremath{(sid, #2)}}{%
\ifthenelse{\isempty{#2}}{\ensuremath{(\mathsf{#1}, sid)}}{%
\ensuremath{(\mathsf{#1}, sid, #2)}}}}

\newcommand{\pot}{\ensuremath{\pi_{OT}}\xspace}
\newcommand{\seed}{s}

\newcommand{\otmsg}{m}
\newcommand{\otp}{\otmsg^{\prime}}
\newcommand{\padm}{p}

\begin{document}

\title{A Framework for Efficient Adaptively Secure Composable Oblivious Transfer in the ROM}
\date{}
\author{
Paulo S. L. M. Barreto \thanks{University of Washington - Tacoma}~~
\and
Bernardo David\thanks{
Tokyo Institute of Technology. Emails:  \texttt{\{bdavid,mario\}@c.titech.ac.jp}.
This work was supported by the Input Output Cryptocurrency Collaborative Research Chair, which has received funding from Input Output HK.}~~
\and 
Rafael Dowsley\thanks{
Aarhus University. Email: \texttt{rafael@cs.au.dk}.
This project has received funding from the European research Council (ERC) under the European Unions's Horizon 2020 research and innovation programme (grant agreement No 669255).
} \and ~~
Kirill Morozov\thanks{Tokyo Institute of Technology}
\and
Anderson C. A. Nascimento \thanks{University of Washington - Tacoma}
}

\maketitle
\pagestyle{plain}



\begin{abstract}
Oblivious Transfer (OT) is a fundamental cryptographic protocol that finds a number of applications, in particular, as an essential building block for two-party and multi-party computation. 
We construct a round-optimal (2 rounds) universally composable (UC) protocol for oblivious transfer secure against active adaptive adversaries from any OW-CPA secure public-key encryption scheme with certain properties in the random oracle model (ROM).
In terms of computation, our protocol only requires the generation of a public/secret-key pair, two encryption operations and one decryption operation, apart from a few calls to the random oracle. In~terms of communication, our protocol only requires the transfer of one public-key, two ciphertexts, and three binary strings of roughly the same size as the message. 
Next, we show how to instantiate our construction under the low noise LPN, McEliece,  QC-MDPC, LWE, and CDH assumptions.
Our instantiations based on the low noise LPN,  McEliece, and QC-MDPC assumptions are the first UC-secure OT protocols based on coding assumptions to achieve: 1) adaptive security, 2) optimal round complexity, 3) low communication and computational complexities. Previous results in this setting only achieved static security and used costly cut-and-choose techniques.
Our instantiation based on CDH achieves adaptive security at the small cost of communicating only two more group elements as compared to the gap-DH based Simplest OT protocol of Chou and Orlandi (Latincrypt 15), which only achieves static security in the ROM.
\end{abstract}


\section{Introduction}\label{sec:intro}


Oblivious transfer (OT) \cite{TR:Rabin81,EveGolLem85} is one of the major protocols within the realm of modern cryptography. It is a fundamental building block for secure two-party and multi-party computation. In this work, we will mainly focus on 1-out-of-2 string oblivious transfer, which is a two-party protocol. Here, the sender (called Alice) inputs two strings $m_0$ and $m_1$, and the receiver (called Bob) inputs a choice bit $c$, and obtains $m_c$ as the output. Bob must not be able to learn $m_{1-c}$, while Alice must not learn $c$. Since oblivious transfer is normally used within other protocols as a primitive, it is desirable to ensure that its security is guaranteed even under concurrent composition, using the universal composability (UC) framework \cite{FOCS:Canetti01}.
Given the possible development of full-scale quantum computers, it is natural to look for protocols that implement oblivious transfer based on assumptions that are not known to be broken by quantum adversaries. 


\subsection{Our contributions}

We propose a framework for obtaining oblivious transfer, which is UC-secure against active adaptive adversaries in the random oracle model.
At the high level, our construction works as follows. We use a public-key encryption (PKE) scheme satisfying the following two properties:

\begin{itemize}
    \item \emph{Property 1 (informal)}: Let the public-key space $\pksp$ form a group with operation denoted by ``$\star$". Then, for the public keys $(\pk_0,\pk_1)$, such that $\pk_0 \star \pk_1 =q$, where $q$ is chosen uniformly at random from $\pksp$, one cannot decrypt both ciphertexts encrypted using $\pk_0$ and $\pk_1$, respectively. In particular, when the public/secret-key pair $(\pk_c,\sk_c)$, $c\in\{0,1\}$, is generated, the above relationship guarantees that $\pk_{1-c}$ that is chosen to satisfy the constraint $\pk_0 \star \pk_1 =q$
    is ``substantially random", so that learning the messages encrypted with $\pk_{1-c}$ is hard.
    
    \item \emph{Property 2 (informal):} $\pk$ obtained using the key generation algorithm is indistinguishable from a random element of $\pksp$. Note that we assume in this work that, in general, not all the elements of $\pksp$ may represent valid public-keys.
\end{itemize}

Now, in our construction, the receiver generates a key pair $(\pk_c,\sk_c)$, queries a random oracle with a random seed value $s$ to obtain $q$, computes $\pk_{1-c}$ such that $\pk_0 \star \pk_1 =q$, and sends $\pk_0$ and $s$ to the sender. The latter obtains $\pk_1$, uses the public keys to encrypt seeds that are used to generate one-time pads (using the random oracle), which in turn she uses to encrypt her respective inputs, and sends the encryptions to the receiver. Intuitively, Property 2 now prevents the sender from learning the choice bit, while Property 1 ensures that the receiver learns at most one of the inputs.

Our construction has the following advantages: 
\begin{itemize}
    \item It can be instantiated with several code-based and lattice-based assumptions, namely low noise LPN, McEliece, QC-MDPC, LWE assumptions. When instantiated with the LPN or McEliece assumptions, our protocol is several orders of magnitude more efficient than previous construction that also achieve UC-security \cite{ICITS:DavNasMul12,CANS:DavDowNas14}. 
   
    \item Our low noise LPN, McEliece and QC-MDPC based instantiations are the first adaptively secure universally composable OT protocols based on these coding assumptions.
   
    \item It can also be instantiated with the CDH assumption. Our framework provides, to the best of our knowledge, the first UC-secure construction of an oblivious transfer protocol based on the standard CDH assumption (Simplest OT~\cite{LC:ChoOrl15} can only be proven assuming gap-DH groups). 
   
    \item Our UC OT protocol based on CDH in the random oracle model with security against adaptive adversaries has basically the same efficiency as the Simplest OT \cite{LC:ChoOrl15}, which is the most efficient UC OT for the static case but can only be proven secure under gap-DH. Our protocol with CDH requires the same number of exponentiations as Simplest OT and only two extra group elements of communication, while being secure against stronger adversaries under a weaker assumption.
\end{itemize}

\paragraph{Concurrent Work} We have recently been made aware of an adaptively secure universally composable oblivious transfer protocol based on the CDH assumption in the ROM proposed by Hauck and Loss~\cite{Hauck17}. Differently from our CDH based construction, which is a corollary of a more general result, Hauck and Loss build on Diffie-Hellman key exchange to construct a protocol based on specific properties of the CDH assumption. 
The authors in  \cite{Hauck17} were interested in obtaining efficient 1-out-of-$n$ oblivious transfer protocols. We, on the other hand, were interested in the case of 1-out-of-2 oblivious transfer, which is the flavor of oblivious transfer that finds the most applications in current protocols for secure multiparty computation. Namely, 1-out-of-2 oblivious transfer serves as basis for the currently most efficient oblivious transfer extension schemes~\cite{C:NNOB12,C:KelOrsSch15}, which cheaply provide large numbers of oblivious transfers used in the currently most efficient commitment~\cite{C:CDDDN16}, two-party computation~\cite{Nielsen17} and multiparty computation~\cite{AC:FKOS15,CCS:KelOrsSch16} protocols. In this scenario, our CDH based instantiation achieves the same computational complexity of~\cite{Hauck17} (requiring 5 modular exponentiations) and requires only 2 extra group elements to be exchanged. In the case of 1-out-of-$n$ oblivious transfer, the protocol of~\cite{Hauck17} presents better computational and communication complexities when compared to our straightforward extension to 1-out-of-$n$, since it only requires a constant number of modular exponentiations and the transfer of symmetric ciphertexts with roughly the same length as the sender's messages. However, we remark our general framework also yields the currently most efficient \emph{post-quantum} constructions under lattice and coding based assumptions for the 1-out-of-2 OT case. 



\subsection{Related Works}

The idea of constructing OT using two public-keys --- the ``pre-computed" one and the ``randomized" one dates back to the CDH-based protocol of Bellare and Micali \cite{C:BelMic89}. It was proven secure in the stand-alone model, and required zero-knowledge proofs. Naor and Pinkas \cite{SODA:NaoPin01}, in particular,\footnote{They have also presented a DDH-based OT protocol in the standard model, but we require the random oracle model as a setup assumption, and hence leave this scheme out of scope of our comparison.} presented an improved and enhanced CDH-based protocol in the random oracle model under the same paradigm, however it was proven secure only in the half-simulation paradigm. It is worth noting that both of the above schemes are tailored for the Diffie-Hellman groups, and hence generalizing them is not trivial. Dowsley \textit{et al.} \cite{DowsleyGMN08} constructed oblivious transfer using the McEliece encryption and the group operation was bitwise exclusive-or of matrices (representing the public-keys). This construction required an expensive cut-and-choose technique, which was leveraged by David \textit{et al.} \cite{ICITS:DavNasMul12} to show the UC-security of this construction. Although Mathew \textit{et al.} \cite{ACISP:MVVR12} showed that the cut-and-choose techniques can be avoided in \cite{DowsleyGMN08} without changing the assumptions, Mathew \textit{et al.} only proved the stand-alone security of their proposal.

Most of the public-key cryptographic schemes that are deployed nowadays have their security based on hardness assumptions coming from number theory, such as factoring and computing discrete logarithms. Likewise, when it comes to (computationally secure) OT protocols, it is possible to build them based on the hardness of factoring~\cite{TR:Rabin81,JC:HalKal12} and on Diffie-Hellman assumptions~\cite{C:BelMic89,SODA:NaoPin01,EC:AieIshRei01,FC:ZLWR13}. In the UC-security setting, OT protocols can be designed under assumptions such as: Decisional Diffie-Hellman (DDH)~\cite{TCC:Garay04,C:PeiVaiWat08}, strong RSA~\cite{TCC:Garay04}, Quadractic Residuosity \cite{C:PeiVaiWat08}, Decisional Linear (DLIN)~\cite{EC:JarShm07,ICISC:DamNieOrl08} and Decisional Composite Residuosity (DCR)~\cite{EC:JarShm07,PKC:CKWZ13}. The Simplest OT protocol of Chou and Orlandi~\cite{LC:ChoOrl15}, which is proven UC-secure against static adversaries in the ROM based on the DDH assumption over a gap-DH group, has the same computational complexity as the CDH-based instantiation of our OT protocol (Section \ref{sec:cdh}). On the one hand, our communication complexity is a bit bigger than theirs; on the other hand, it provides security against adaptive adversaries under a weaker assumption. We would like to emphasize that although our CDH-based protocol is somewhat similar to the basic protocol by Naor and Pinkas \cite{SODA:NaoPin01}, there are the following two crucial differences: 1) In their protocol, the sender chooses randomness that is used to ``randomize" the receiver's keys, while in our protocol, that randomness comes from applying the random oracle to a seed chosen by the receiver himself; 2) Our protocol uses a different encryption method for the sender's messages, in particular, the ElGamal encryption is employed, which is not the case in their protocol. It is exactly those differences that allow us to leverage stronger security guarantees (UC-security) as compared to the Naor-Pinkas protocol (half-simulation security).

It has been known for more than two decades that Shor's algorithm~\cite{FOCS:Shor94} makes 
factoring and computing discrete logarithms easy for quantum computers. Therefore, an important pending problem concerns designing post-quantum OT protocols. 
One solution is to rely on statistically secure protocols (i.e., those not depending on any computational assumption and thus secure even if full-scale quantum computers become reality) based on assumptions such as the existence of noisy channels \cite{FOCS:CreKil88,EC:Crepeau97,EC:DamKilSal99,ISIT:SteWol02,SCN:CreMorWol04,IEEEIT:NasWin08,IEEEIT:PDMN11,AhlCsi13,IEEEIT:DowNas17}, pre-distributed correlated data \cite{STOC:Beaver97,Rivest99,IEEEIT:NasWin08}, cryptogates \cite{STOC:Kilian00,C:BeiMalMic99}, the bounded storage model~\cite{FOCS:CacCreMar98,TCC:DHRS04,ISIT:DowLacNas14,DowLacNas15} and on hardware tokens~\cite{TCC:DotKraMul11,ICITS:DowMulNil15}. Nonetheless, these constructions (except the ones based on a trusted initializer) are rather impractical. Thus, it seems reasonable to focus on obtaining OT protocols based on computational problems that are believed to be hard even for quantum computers: such as, for instance, the Learning from Parity with Noise (LPN), the Learning with Errors (LWE), and the McEliece problems. 

The LPN problem essentially states that given a system of binary linear equations, in which the outputs are disturbed by some noise, it is difficult to determine the solution. It is very simple to generate LPN samples, however finding the solution seems to be very hard \cite{BluKalWas03,Lyubashevsky05,SCN:LevFou06,EPRINT:Kirchner11}. Therefore it is an attractive assumption that has been widely used for symmetric cryptographic primitives \cite{AC:HopBlu01,C:JueWei05,ICALP:GilRobSeu08,C:ACPS09,JC:KatShiSmi10,EC:KPCJV11}. It has seen far less usage in asymmetric cryptographic primitives. However, public-key encryption schemes \cite{FOCS:Alekhnovich03,EPRINT:DamPar12,AC:DotMulNas12} were designed based on the LPN variant introduced by Alekhnovich \cite{FOCS:Alekhnovich03}, which has low noise but only provides a linear amount of samples. An OT protocol was also designed based on this variant of LPN \cite{CANS:DavDowNas14}. However, this protocol is based on cut-and-choose techniques, which require a number of key generation, encryption and decryption operations linear in the cut-and-choose statistical security parameter. On the other hand, the instantiation of our OT protocol using this LPN variant (Section \ref{sec:lpnpke}) is far more efficient.

The LWE problem is a generalization of the LPN problem that was introduced by Oded Regev \cite{STOC:Regev05}. It is as hard to solve as some lattice problems and is one of the most versatile assumptions used in cryptography (e.g., \cite{STOC:Regev05,STOC:GenPeiVai08,EC:CHKP10,FOCS:BraVai11,EC:BanPeiRos12,EC:CanChe17}). Peikert \textit{et al.} \cite{C:PeiVaiWat08} proposed a framework for realizing efficient,
round-optimal UC-secure OT protocols that can be instantiated based on LWE. Their framework works in the CRS model. The main bulding block of this framework is a dual-mode encryption, which Peikert \textit{et al.} built from a multi-bit version of Regev's IND-CPA cryptosystem. Their final instantiation requires two key generations, two encryptions and one decryption with the basic IND-CPA cryptosystem. We employ the same basic cryptosystem in our LWE instantiation but perform one less key generation.

McEliece~\cite{DSN:McEliece78} introduced a public-key encryption scheme based on hardness of the syndrome  decoding problem. The cryptosystem proposed by Niederreiter~\cite{PCIT:Niederreiter86} is its dual. Later on, IND-CPA-secure~\cite{WCC:NIKM07,DCC:NIKM08} and IND-CCA2-secure~\cite{RSA:DowMulNas09,PKC:FGKRS10,IEEEIT:DDMN12} variants of these schemes were introduced. Both stand-alone~\cite{ICITS:DGMN08,IEICE:DGMN12,ACISP:MVVR12}
as well as fully-simulatable~\cite{DN2011,DNS2012} and UC-secure~\cite{ICITS:DavNasMul12} OT protocols can be built using these cryptosystems. As in the case of low noise LPN, most of these protocols (and all UC-secure ones) are based on cut-and-choose techniques, which require a number of key generation, encryption and decryption operations linear in the cut-and-choose statistical security parameter. At the same time, our construction uses a small constant number of these operations. All the above advanced constructions assume pseudorandomness of the public-keys of the McEliece and Niederreiter PKE's. Our OT protocol based on the McEliece encryption scheme (Section \ref{sec:mepke}) is far more efficient than the ones using cut-and-choose techniques.

A number of adaptively secure universally composable oblivious transfer protocols have been proposed in current literature~\cite{STOC:CLOS02,ACNS:BlaChe15,ACNS:BlaCheGer17,PKC:CKWZ13,C:GarWicZho09,AC:BlaChe16,TCC:CDMW09}. Most of these protocols are on the CRS model and on erasures in order to achieve adaptive security. Moreover, these protocols require more than 2 rounds and rely either on complex zero-knowledge proof techniques or on adaptively secure universally composable commitments as central building blocks of their constructions. On the other hand, our generic construction does not assume secure erasures and requires solely a simple OW-CPA secure cryptosystem and can be executed in 2 rounds, which is optimal. Moreover, by avoiding heavy primitives such as adaptively secure UC commitments and zero-knowledge proofs, we achieve much better concrete computational and communication complexities. Besides this stark contrast in efficiency, our results show that adaptively secure UC OT can be achieved under much weaker assumptions in the random oracle model in comparison to the CRS model.

It is a well-known fact that UC-secure OT protocols require some setup assumption \cite{C:CanFis01}. Our protocols use the random oracle model as such. Alternative setup assumptions that can be used to obtain UC-secure OT protocols include the common reference string (CRS) model \cite{STOC:CLOS02,TCC:Garay04,C:PeiVaiWat08}, the public-key infrastructure model \cite{C:DamNie03}, the existence of noisy channels \cite{SBSEG:DMN08,JIT:DGMN13}, and tamper-proof hardware~\cite{EC:Katz07,TCC:DotKraMul11,ICITS:DowMulNil15}.

\section{Preliminaries}\label{sec:pre}

We denote by $\secpar$ the security parameter. Let $y \getsr F(x)$ denote running the randomized algorithm $F$ with input $x$ and random coins, and obtaining the output $y$. Similarly, $y \gets F(x)$ is used for a deterministic algorithm. For a set $\calX$, let $x \getsr \calX$ denote $x$ chosen uniformly at random from $\calX$; and for a distribution $\calY$, let $y \getsr \calY$ denote $y$ sampled according to the distribution $\calY$. Let $\Ft$ denote the finite field with 2 elements. For $A, B \in \Ft^{m \times n}$, $A \xor B$ denotes their element-wise exclusive-or. For a parameter $\rho$, 
$\chi_\rho$ denotes the Bernoulli distribution that outputs 1 with probability $\rho$. Let $\HW(e)$ denote the Hamming weight of a vector $e$, i.e., the number of its non-zero positions. 
We will denote by $\mathsf{negl}(\secpar)$ the set of negligible functions of $\secpar$. We abbreviate \emph{probabilistic polynomial time} as PPT.

\subsection{Encryption Schemes}
\label{sec:enc}

The main building block used in our OT protocol is a public-key encryption scheme. Such a scheme, denoted as $\pke$, has public-key $\pksp$, secret-key $\sksp$, message $\messp$, randomness $\rndsp$ and ciphertext $\ciphsp$ spaces that are functions of the security parameter $\secpar$, and 
consists of the following three algorithms $\kg$, $\enc$, $\dec$: 

\begin{itemize}
\item The PPT key generation algorithm $\kg$ takes as input the security parameter 
$1^\secpar$ and outputs a pair of public $\pk \in \pksp$ and secret $\sk \in \sksp$ keys.

\item The PPT encryption algorithm $\enc$ takes as input a public-key $\pk \in \pksp$, a message $\mes \in \messp$ and randomness $\rnd \in \rndsp$ and outputs a ciphertext $\ciph \in \ciphsp$. We denote this operation by $\enc(\pk,\mes,\rnd)$. When $\rnd$ is not explcitly given as input, it is assumed to be sampled uniformly at random from $\rndsp$.

\item The (deterministic) decryption algorithm $\dec$ takes as input the secret-key $\sk \in \sksp$ and a ciphertext $\ciph \in \ciphsp$ and outputs either a message $\mes \in \messp$ or an error symbol $\perp$. For $(\pk,\sk)\getsr \kg(1^\secpar)$, any $\mes \in \messp$, and $c \getsr \enc(\pk,\mes)$, it should hold that $\dec(\sk,\ciph)=\mes$ with overwhelming probability over the randomness used by the algorithms.
\end{itemize}

We should emphasize that for some encryption schemes not all $\widetilde{\pk} \in \pksp$ are ``valid'' in the sense of being a possible output of $\kg$.
The same holds for $\widetilde{\ciph} \in \ciphsp$ in relation to $\enc$ and all possible coins and messages. 

Next we define the notion of one-wayness against chosen-plaintext attacks (OW-CPA).

\begin{definition}[OW-CPA security] 
A $\pke$ is OW-CPA secure if for every PPT adversary $\adv$, and for $(\pk,\sk) \getsr \kg(1^\secpar)$, $\mes \getsr \messp$ and $\ciph \getsr \enc(\pk,\mes)$, 
it holds that 
\[
\Pr[\adv(\pk,\ciph) = \mes ] \in \mathsf{negl}(\secpar).
\]
\end{definition}

Our OT constructions use as a building block a $\pke$ that satisfies a variant of the OW-CPA security notion: informally, two random messages are encrypted under two different public-keys, one of which can be chosen by the adversary (but he does not have total control over both public-keys). His goal is then to recover both messages and this should be difficult. Formally, this property is captured by the following definition.

\begin{property}\label{ass:dowa}
Consider the public-key encryption scheme $\pke$ and the security parameter $\secpar$. It is assumed that $\pksp$ forms a group with operation denoted by ``$\star$". For every PPT  two-stage adversary $\adv=(\adv_1,\adv_2)$ running the following experiment: 
\begin{tabbing}
1234567\=123\=123\=123\=\kill
\> $q \getsr \pksp$\\
\> $(\pk_1,\pk_2,\st) \getsr \adv_1(q) \text{ such that } 
\pk_1, \pk_2 \in \pksp \text{ and }
\pk_1 \star \pk_2 =q$\\
\> $\mes_i  \getsr \messp \text{ for } i=1,2$\\
\> $\ciph_i \getsr \enc(\pk_i,\mes_i) \text{ for } i=1,2$\\
\> $(\widetilde{\mes_1},\widetilde{\mes_2}) \getsr \adv_1(\ciph_1,\ciph_2,\st)$
\end{tabbing}
it holds that
\[
\Pr[(\widetilde{\mes_1},\widetilde{\mes_2}) = (\mes_1,\mes_2) ] \in \mathsf{negl}(\secpar).
\]
\end{property}

We also need a property about indistinguishability of a public-key generated using $\kg$ and an element sampled uniformly at random from $\pksp$.

\begin{property}\label{ass:keys}
Consider the public-key encryption scheme $\pke$ and the security parameter $\secpar$. Let $(\pk,\sk) \getsr \kg(1^\secpar)$ and $\pk' \getsr \pksp$. For every PPT distinguisher $\adv$, it holds that
\[
| \Pr[\adv(\pk) = 1 ] - \Pr[\adv(\pk') = 1 ] | \in \mathsf{negl}(\secpar).
\]
\end{property}

In Section \ref{sec:inst}, we describe some cryptosystems for which Properties \ref{ass:dowa} and \ref{ass:keys} are believed to hold and thus they can be used to instantiate our OT protocol.

\subsection{Universal Composability}
We prove our protocols secure in the Universal Composability (UC) framework introduced by Canetti in~\cite{FOCS:Canetti01}. 
In this section, we present a brief description of the UC framework originally given in~\cite{PKC:CDDGNT15} and refer interested readers to~\cite{FOCS:Canetti01} for further details.
In this framework, protocol security is analyzed under the real-world/ideal-world paradigm, \textit{i.e.}, by comparing the real world execution of a protocol with an ideal world interaction with the primitive that it implements. The model includes a \textit{composition theorem}, that basically states that UC secure protocols can be arbitrarily composed with each other without any security compromises. This desirable property not only allows UC secure protocols to effectively serve as building blocks for complex applications but also guarantees security in practical environments, where several protocols (or individual instances of protocols) are executed in parallel, such as the Internet. 

In the UC framework, the entities involved in both the real and ideal world executions are modeled as PPT  Interactive Turing Machines (ITM) that receive and deliver messages through their input and output tapes, respectively. In the ideal world execution, dummy 
parties (possibly controlled by an ideal adversary $\suc$ referred to as the \textit{simulator}) interact directly with 
the ideal functionality $\Fuc$, which works as a trusted third party that computes the desired 
primitive. In the real world execution, several parties (possibly corrupted 
by a real world adversary $\adv$) interact with each other by means of a protocol $\pi$ that 
realizes the ideal functionality. The real and ideal executions are controlled by the \textit{environment} $\env$, 
an entity that delivers inputs and reads the outputs 
of the individual parties, the adversary $\adv$ and the simulator $\suc$. After a real or ideal execution, $\env$ outputs a bit, which is considered as the output of the execution. The rationale behind this framework lies in showing that the 
environment  $\env$ (that represents everything that happens outside of the protocol
execution) is not able to efficiently distinguish between the real and ideal executions, thus implying that the real world protocol is as secure as 
the ideal functionality.

We denote by $\mathsf{REAL}_{\pi, \adv, \env}(\secpar,z,\bar{r})$ the output of the environment $\env$ in the real-world execution of a protocol $\pi$ between $n$ parties with an adversary $\adv$ under security parameter $\secpar$, input $z$ and randomness $\bar{r}=(r_{\env},r_{\adv},r_{P_1},\ldots,r_{P_n})$, where $(z,r_{\env})$, $r_{\adv}$ and $r_{P_i}$ are respectively related to $\env$, $\adv$ and party $i$. Analogously, we denote by $\mathsf{IDEAL}_{\Fuc, \suc, \env}(\secpar,z,\bar{r})$ the output of the environment in the ideal interaction between the simulator $\suc$ and the ideal functionality $\Fuc$ under security parameter $\secpar$, input $z$ and randomness $\bar{r}=(r_{\env},r_{\suc},r_{\Fuc})$, where $(z,r_{\env})$, $r_{\suc}$ and $r_{\Fuc}$ are respectively related to $\env$, $\suc$ and $\Fuc$.
The real world execution and the ideal executions are respectively represented by the ensembles 
$\mathsf{REAL}_{\pi, \adv, \env}=\{\mathsf{REAL}_{\pi, \adv, \env}(\secpar,z,\bar{r})\}_{\secpar \in \N}$ and 
$\mathsf{IDEAL}_{\Fuc, \suc, \env}=\{\mathsf{IDEAL}_{\Fuc, \suc, \env}(\secpar,z,\bar{r})\}_{\secpar \in \N}$ with $z \in \bits^{*}$ and  a uniformly chosen $\bar{r}$.

In addition to these two models of computation, the UC framework also considers the $\mathcal{G}$-hybrid world,  where the computation proceeds as in the real-world with the additional assumption  that the parties have access to an auxiliary ideal functionality $\mathcal{G}$. In this model, honest parties do not communicate with the ideal functionality directly, but instead the adversary delivers all the messages to and from the ideal functionality. We consider the communication channels to be ideally authenticated, so that the adversary may read but not modify these messages. Unlike messages exchanged between parties, which can be read by the adversary, the messages exchanged between parties and the ideal functionality are divided into a \textit{public header} and a \textit{private header}. The public header can be read by the adversary and contains non-sensitive information (such as session identifiers, type of message, sender and receiver). On the other hand, the private header cannot be read by the 
adversary and contains information such as the parties' private inputs. We denote the ensemble of environment outputs that represents 
an execution of a protocol $\pi$ in a $\mathcal{G}$-hybrid model as $\mathsf{HYBRID}^{\mathcal{G}}_{\pi, \adv, \env}$ (defined analogously to $\mathsf{REAL}_{\pi, \adv, \env}$).  UC security is then formally defined as:  

\begin{definition}
An n-party ($n \in \N$) protocol $\pi$ is said to UC-realize an ideal functionality $\mathcal{F}$ in the $\mathcal{G}$-hybrid model if, for every adversary $\mathcal{A}$, there exists a  simulator $\mathcal{S}$ such that, for every environment $\mathcal{Z}$, the following relation holds: $$\mathsf{IDEAL}_{\Fuc, \suc, \env}\approx \mathsf{HYBRID}^{\mathcal{G}}_{\pi, \adv, \env}$$
\end{definition}
We say that a protocol is \emph{statistically secure}, if the same holds for all $\env$ with 
unbounded computing power.

\subsubsection{Adversarial Model:} We consider a malicious adversary, which can deviate from the prescribed protocol in an arbitrary way. We call the adversary \emph{static}, if he has to corrupt parties before execution starts and the corrupted (or honest) parties remain as such throughout the execution. We call the adversary \emph{adaptive}, if he is able corrupt parties at any point in the protocol execution, and even after it.

\subsubsection{Oblivious Transfer Ideal Functionality:} The basic 
1-out-of-2 string oblivious transfer functionality $\Fot$ as defined 
in~\cite{STOC:CLOS02} is presented in Figure \ref{fig:fot}.

\begin{boxfig}{Functionality \Fot.}{fig:fot}
\begin{center}
\textbf{Functionality } \Fot.
\end{center}

\Fot interacts with a sender $\snd$ and a receiver $\rec$. The lengths of the strings $\lenf$ is fixed and known to both parties. 
\Fot proceeds as follows:
\begin{itemize}
\item Upon receiving a message $\msg{sender}{\vec{x}_0, \vec{x}_1 }$ from $\snd$, where each $\vec{x}_i \in \{0, 1\}^{\lenf}$, store the tuple $(sid, \vec{x}_0 , \vec{x}_1)$. Ignore further messages from $\snd$ with the same $sid$. 

\item Upon receiving a message $\msg{receiver}{c}$ from $\rec$, where $c\in \{0,1\}$, check if a tuple $(sid, \vec{x}_0, \vec{x}_1)$ was recorded. If yes, send $\msg{received}{\vec{x}_{c}}$ to $\rec$ and $\msg{received}{}$ to $\snd$ and halt. Otherwise, send nothing to $\rec$, but continue running.
\end{itemize}
\end{boxfig}

\subsubsection{Setup Assumptions:} Our constructions rely on the random oracle model~\cite{CCS:BelRog93}, which can be modelled in the UC framework as the $\Fro$-hybrid model. The random oracle functionality $\Fro$ is presented in Figure~\ref{fig:fro}. Our construction will actually use two instances of $\Fro$: $\Frok$ with range $\pksp$ and $\Frop$ with range $\{0,1\}^{\lenf}$.


\begin{boxfig}{Functionality \Fro.}{fig:fro}
\begin{center}
\textbf{Functionality } \Fro
\end{center}

\Fro is parameterized by a range $\calD$. $\Fro$ keeps a list $L$ of pairs of values, which is initially empty, and proceeds as follows:

\begin{itemize}

\item Upon receiving a value $\msg{}{m}$ from a party $P_i$ or from $\suc$, if there is a pair $(m,\hat{h})$ in the list $L$, set $h = \hat{h}$. Otherwise, choose $h \getsr \calD$ and store the pair $(m,h)$ in $L$.
Reply to the activating machine  with $\msg{}{h}$.
\end{itemize}

\end{boxfig}

\section{Oblivious Transfer Protocol}\label{sec:protocol}

In this section, we introduce our 1-out-of-2 OT protocol and prove its UC security against static malicious adversaries in the \Fro-hybrid model (\textit{i.e.}, the random oracle model). Our protocol uses as a building block a public-key encryption scheme that satisfies Properties \ref{ass:dowa} and \ref{ass:keys} (defined in Section \ref{sec:enc}). The high-level idea is that $\rec$ picks two public-keys $\pk_0,\pk_1$ such that he only knows the secret key corresponding to $\pk_c$ (where $c$ is his choice bit) and hands them to $\snd$. She then uses the two public-keys to transmit the messages in an encrypted way, so that $\rec$ can only recover the message, for which he knows the secret-key $\sk_c$. 

A crucial point in such schemes is making sure that $\rec$ is only able to decrypt one of the messages. In order to enforce this property, our protocol relies on Property \ref{ass:dowa} and uses the random oracle to force the element $q$ to be chosen uniformly at random from $\pksp$. After generating the pair of secret and public key $(\sk_c,\pk_c)$,
$\rec$ samples a seed $\seed$, queries the random oracle $\Frok$ to obtain $q$ and computes $\pk_{1-c}$ such that $\pk_0 \star \pk_{1} =q$. $\rec$ then hands the public-key $\pk_0$ and the seed $\seed$ to $\snd$, enabling her to also compute $\pk_1$.  Since the public-keys are indistinguishable according to Property \ref{ass:keys}, $\snd$ learns nothing about $\rec$'s choice bit.
Next, $\snd$ picks two uniformly random strings $\padm_0,\padm_1$, queries them to the random oracle $\Frop$ obtaining $\padm_0^{\prime},\padm_1^{\prime}$ as response,
and computes one-time pad encryptions of her messages $\otmsg_0,\otmsg_1$ as $\otp_0 = \otmsg_0 \oplus \padm_0^{\prime}$ and $\otp_1 = \otmsg_1 \oplus\padm_1^{\prime}$. $\snd$ also encrypts $\padm_0$ and $\padm_1$ under $\pk_0$ and $\pk_1$, respectively, to obtain $\ciph_0$ and $\ciph_1$. $\snd$ sends $(\otp_1,\otp_2,\ciph_0,\ciph_1)$ to $\rec$. $\rec$ can use $\sk_c$ to decrypt $\ciph_c$ obtaining $\padm_c$. He then queries $\padm_c$ to the random oracle $\Frop$ obtaining $\padm_c^{\prime}$ as response,
and retrieves $\otmsg_c=\otp_c \oplus \padm_c^{\prime}$. Due to Property \ref{ass:dowa}, $\rec$ will not be able to recover $\padm_{1-c}$ in order to query it to the random oracle and to decrypt $\otp_{1-c}$. Therefore, the security for $\snd$ is also guaranteed.

Our protocol $\pot$ is described in Figure~\ref{fig:pot}. It can be instantiated using different public-key encryption schemes as described in Section \ref{sec:inst}.

\begin{boxfig}{Protocol \pot}{fig:pot}
\begin{center}
\textbf{Protocol } \pot
\end{center}

Let $\pke$ be a public-key encryption scheme that satisfies Properties \ref{ass:dowa} and \ref{ass:keys}, and $\secpar$ be the security parameter. Protocol \pot is executed between $\snd$ with inputs $\otmsg_0,\otmsg_1 \in \{0,1\}^{\lenf}$ and $\rec$ with input $c \in \{0,1\}$. They interact with each other and with two instances of the random oracle ideal functionality $\Fro$ ($\Frok$ with range $\pksp$ and $\Frop$ with range $\{0,1\}^{\lenf}$), proceeding as follows:

\begin{enumerate}
\item $\rec$ generates a pair of keys $(\pk_c,\sk_c) \getsr \kg(1^{\secpar})$. He samples a random string $\seed \getsr \{0,1\}^{\secpar}$ and sends $\msg{}{\seed}$ to $\Frok$, obtaining $\msg{}{q}$ as answer. $\rec$ computes 
$\pk_{1-c}$ such that $\pk_0 \star \pk_{1} =q$ and sends $(\seed,\pk_0)$ to $\snd$.

\item Upon receiving $(\seed,\pk_0)$, $\snd$:
    \begin{enumerate}
        \item Queries $\Frok$ with $\msg{}{\seed}$, obtaining $\msg{}{q}$ in response. 
        
        \item Computes $\pk_1$ such that $\pk_0 \star \pk_{1} =q$.
            
        \item Samples $\padm_0, \padm_1  \getsr \{0,1\}^{\secpar}$ and queries $\Frop$ with $\msg{}{\padm_0}$ and $\msg{}{\padm_1}$, obtaining $\msg{}{\padm_0^{\prime}}$ and $\msg{}{\padm_1^{\prime}}$ as answers.  
        
        \item  Computes $\otp_0=\padm_0^{\prime} \oplus \otmsg_0$, $\otp_1=\padm_1^{\prime} \oplus \otmsg_1$, $\ciph_0 \getsr \enc(\pk_0,\padm_0)$ and $\ciph_1 \getsr \enc(\pk_1,\padm_1)$. 
        
        \item Sends $(\otp_0,\otp_1,\ciph_0,\ciph_1)$ to $\rec$ and halts.
    \end{enumerate}

\item $\rec$ computes $\padm_c \gets \dec(\sk_c,\ciph_c)$. If decryption fails, $\rec$ outputs $\otmsg_c \getsr \{0,1\}^{\lenf}$ and halts. If decryption succeeds $\rec$ queries $\Frop$ with $\msg{}{\padm_c}$, obtaining $\msg{}{\padm_c^{\prime}}$ as answer. He then outputs $\otmsg_c = \otp_c \oplus \padm_c^{\prime}$. 

\end{enumerate}

\end{boxfig}

\subsection{Security Analysis}

We now formally state the security of $\pot$.

\begin{theorem}\label{th:potucme}
Let $\pke$ be a public-key encryption scheme that satisfies Properties \ref{ass:dowa} and \ref{ass:keys}. When instantiated with $\pke$, the protocol $\pot$ UC-realizes the functionality $\Fot$ against static adversaries in the $\Fro$-hybrid model.
\end{theorem}

\begin{proof}
In order to prove the security of $\pot$, we must construct a simulator $\suc$ such that no environment $\env$ can distinguish between interactions with an adversary $\adv$ in the real world and with $\suc$ in the ideal world. For the sake of clarity, we will describe the simulator $\suc$ separately for the case where only $\rec$ is corrupted and the case where only $\snd$ is corrupted. In both cases, $\suc$ writes all the messages received from $\env$ in $\adv$'s input tape, simulating $\adv$'s environment. Also, $\suc$ writes all messages from $\adv$'s output tape to its own output tape, forwarding them to $\env$.
Notice that simulating the cases where both Alice and $\rec$ are honest or corrupted is trivial. If both are corrupted, $\suc$ simply runs $\adv$ internally. In this case, $\adv$ will generate the messages from both corrupted parties. If neither $\snd$ nor $\rec$ are corrupted, $\suc$ runs the protocol between honest $\snd$ and $\rec$ internally on the inputs provided by $\env$ and all messages are delivered to $\adv$. 
As for the correctness, notice that $\padm_c$ is encrypted under public key $\pk_c$, for which $\rec$ knows the corresponding secret key $\sk_c$. Hence, $\rec$ can successfully recover $\padm_c$ from $\ciph_c$ and use it to obtain $\padm_c^{\prime}$ from $\Fro$. This enables $\rec$ to retrieve $\otmsg_c$ from $\otp_c$ by computing $\otmsg_c = \otp_c \oplus \padm_c^{\prime}$.

\paragraph{Simulator for a corrupted Bob:} In the case where only $\rec$ is corrupted, the simulator $\suc$ interacts with $\Fot$ and an internal copy $\adv$ of the real world adversary ($\suc$ acts as $\snd$ in this internal simulated execution of the protocol). Additionally, $\suc$ also plays the role of $\Frok$ and $\Frop$ to $\adv$.
The goal of the simulator is to extract $\rec$'s choice bit in order to request the correct message from $\Fot$. In order to do so, $\suc$ first executes the same steps of an honest $\snd$ in $\pot$ with the difference that it picks uniformly random values for $\otp_0,\otp_1$ instead of computing them. Since $\suc$ plays the role of $\Frop$ to $\adv$, it learns $\rec$'s choice bit when it receives a query $\msg{}{\padm_c}$ from $\adv$ (originally meant for $\Frop$ in the real world). Since $\adv$ can only query the random oracle $\Frok$ in a polynomial number of points to obtain $q$, an $\adv$ that queries $\Frop$ with $\padm_{1-c}$ (thus tricking $\suc$ into extracting the wrong choice bit) can be trivially used to break Property \ref{ass:dowa}. Knowing $c$, $\suc$ obtains $\otmsg_c$ from $\Fot$ and answers $\adv$'s query to $\Frop$ with $\msg{}{\padm_{c}^{\prime}}$ 
such that $\padm_c^{\prime} = \otp_c \oplus \otmsg_c$. The simulator $\suc$ for the case where only $\rec$ is corrupted is presented in Figure~\ref{fig:sucrec}.

\begin{boxfig}{Simulator $\suc$ for the case where only $\rec$ is corrupted.}{fig:sucrec}
\begin{center}
\textbf{Simulator } $\suc$ (Corrupted $\rec$)
\end{center}

Let $\lenf$ be the length of the messages and $\secpar$ be the security parameter.
The simulator $\suc$ interacts with an environment $\env$, functionality $\Fot$ and an internal copy $\adv$ of the adversary that corrupts only $\rec$, proceeding as follows:

\begin{enumerate}
\item $\suc$ simulates the answers to the random oracle queries from $\adv$ exactly as $\Frok$ and $\Frop$ would (and stores the lists of queries/answers), except when stated otherwise in $\suc$'s description.

\item Upon receiving $(\seed,\pk_0)$ from $\adv$, $\suc$ proceeds as follows:

    \begin{enumerate}
        \item Checks if there is a pair $(\seed,q)$ in the simulated list of $\Frok$. If not, $\suc$ samples $q \getsr \pksp$ and stores $(\seed,q)$ in it. $\suc$ computes $\pk_1$ such that $\pk_0 \star \pk_1 =q$.
        
        \item Samples $\otp_0,\otp_1 \getsr \{0,1\}^{\lenf}$, $\padm_0, \padm_1 \getsr \{0,1\}^{\secpar}$. If $\padm_0$ or $\padm_1$ has already been queried to $\Frop$, $\suc$ aborts. Otherwise, $\suc$ computes $\ciph_0 \getsr \enc(\pk_0,\padm_0)$ and $\ciph_1 \getsr \enc(\pk_1,\padm_1)$, and sends $(\otp_0,\otp_1,\ciph_0,\ciph_1)$ to $\adv$.
    \end{enumerate}
    
\item $\suc$ resumes answering the random oracle queries from $\adv$ as above, except for the queries $\padm_0$ and $\padm_1$ to $\Frop$. Upon receiving such query $\padm_c$, $\suc$ queries $\Fot$ with $\msg{receiver}{c}$, receiving $\msg{received}{\otmsg_{c}}$ in response. $\suc$ answers $\adv$'s query with the value $\otp_c \oplus \otmsg_{c}$. After this point, if $\adv$ queries $\padm_{1-c}$ to $\Frop$, $\suc$ aborts.
When $\adv$ halts, $\suc$ also halts and outputs whatever $\adv$ outputs.

\end{enumerate}

\end{boxfig}

Notice that unless the simulator $\suc$ aborts, it perfectly emulates the execution of the real protocol for $\adv$. Hence, an environment can distinguish the real world execution from the ideal world simulation only if: (1) $\padm_0$ or $\padm_1$ is queried to $\Frop$ before the simulator sends the encrypted messages; or (2) $\adv$ queries both $\padm_0$ or $\padm_1$ to $\Frop$. The first event can only happen with negligible probability as $\padm_0$ and $\padm_1$ are uniformly random strings of length $\secpar$ and the PPT adversary $\adv$ can only make a polynomial number of queries to the random oracle. The probability of the second event is polynomially related with the probability of breaking Property \ref{ass:dowa}.
Thus, a PPT environment $\env$ cannot distinguish an interaction with $\suc$ in the ideal world from an interaction with $\adv$ in the real world except with negligible probability.

\paragraph{Simulator for a corrupted Alice:} 

In the case where only $\snd$ is corrupted, the simulator $\suc$ interacts with $\Fot$ and an internal copy $\adv$ of the real world adversary
($\suc$ acts as $\rec$ in this internal simulated execution of the protocol). Additionally, $\suc$ also plays the role of $\Frok$ and $\Frop$ to $\adv$. The goal of the simulator is to extract both messages $\otmsg_0,\otmsg_1$ of the receiver in order to deliver them to $\Fot$. In order to do so, $\suc$ has to trick $\adv$ into accepting two public-keys $\pk_0,\pk_1$ for which $\suc$ knows the corresponding secret-keys $\sk_0,\sk_1$. $\suc$ generates two secret and public-key pairs $(\pk_0,\sk_0) \getsr \kg(1^{\secpar})$ and $(\pk_1,\sk_1) \getsr \kg(1^{\secpar})$. Additionally, $\suc$ generates a random seed $\seed \getsr \{0,1\}^{\secpar}$ and sends $(\seed,\pk_0)$ to $\rec$. 
When $\adv$ queries $\Frok$ with $\msg{}{\seed}$, $\suc$ answers with $\msg{}{\pk_0 \oplus \pk_1}$, thus making $\adv$ fix $\pk_1$ for which $\suc$ knows the corresponding secret-key $\sk_1$. Upon receiving $(\otp_0,\otp_1,\ciph_0,\ciph_1)$ from $\adv$, $\suc$ uses $\sk_0$ and $\sk_1$ to decrypt both $\ciph_0$ and $\ciph_1$ and obtain $\padm_0$ and $\padm_1$ (respectively), thus enabling it to recover both messages of $\snd$. If the decryption of $\ciph_i$ fails, $\suc$ samples a random $\otmsg_i \getsr \{0,1\}^{\lenf}$. The simulator $\suc$ for the case where only $\snd$ is corrupted is presented in Figure~\ref{fig:sucsnd}.

\begin{boxfig}{Simulator $\suc$ for the case where only $\snd$ is corrupted.}{fig:sucsnd}
\begin{center}
\textbf{Simulator } $\suc$ (Corrupted $\snd$)
\end{center}

Let $\lenf$ be the length of messages and $\secpar$ be a security parameter.
Simulator $\suc$ interacts with an environment $\env$, functionality $\Fot$ and an internal copy $\adv$ of the adversary that corrupts only Alice, proceeding as follows:

\begin{enumerate}

\item $\suc$ simulates the answers to the random oracle queries from $\adv$ exactly as $\Frok$ and $\Frop$ would (and stores the lists of queries/answers), except when stated otherwise in $\suc$'s description.

\item $\suc$ generates two secret and public-key pairs $(\pk_0,\sk_0) \getsr \kg(1^{\secpar})$ and $(\pk_1,\sk_1) \getsr \kg(1^{\secpar})$. 
$\suc$ samples $\seed \getsr \{0,1\}^{\secpar}$ and aborts if $\seed$ has been already queried to $\Frok$. $\suc$ sends $(\seed,\pk_0)$ to $\adv$.

\item When $\adv$ queries $\Frok$ with $\msg{}{\seed}$, $\suc$ answers with $\msg{}{\pk_0 \oplus \pk_1}$.

\item Upon receiving $(\otp_0,\otp_1,\ciph_0,\ciph_1)$, $\suc$ proceeds as follows:
    \begin{enumerate}
        \item Decrypts $\ciph_0$ and $\ciph_1$, getting $\padm_0 \gets \dec(\sk_0,\ciph_0)$ and $\padm_1 \gets \dec(\sk_1,\ciph_1)$. If the decryption of $\ciph_i$ failed, $\suc$ samples $\otmsg_i \getsr \{0,1\}^{\lenf}$. Otherwise, $\suc$ recovers $\padm_i^{\prime}$ from the list $(\padm_i,\padm_i^{\prime})$ stored for $\Frop$ (or samples $\padm_i^{\prime} \getsr \bits^{\lenf}$ and stores the new pair if such a pair does not exist yet) and computes both messages $\otmsg_0 = \otp_0 \oplus \padm_0^{\prime}$ and $\otmsg_1 = \otp_0 \oplus \padm_0^{\prime}$.

        \item  $\suc$ sends $\msg{sender}{\otmsg_0,\otmsg_1}$ to $\Fot$.
    \end{enumerate}

 \item When $\adv$ halts, $\suc$ also halts and outputs whatever $\adv$ outputs.

\end{enumerate}

\end{boxfig}

The only points where the simulation differs from the real world protocol execution are in the sampling of two public-keys $\pk_0,\pk_1$ for which $\suc$ knows the corresponding secret-keys and in the answer $\msg{}{\pk_0 \oplus \pk_1}$ to the query $\msg{}{\seed}$ from $\adv$ to (the simulated) $\Frok$. 
Notice that due to Property \ref{ass:keys}, a public-key outputted by $\kg$ is indistinguishable from a key sampled uniformly at random from $\pksp$. Hence, the choice of public-keys in the simulation is indistinguishable from that of the real world protocol execution. The only way the simulation can still fail is if $\adv$ queries $\Frok$ with $\msg{}{\seed}$ before $\pk_0,\pk_1$ are chosen by $\suc$, which happens with negligible probability. Thus, a PPT environment $\env$ cannot distinguish an interaction with $\suc$ in the ideal world from an interaction with $\adv$ in the real world except with negligible probability.
\end{proof}

\subsection{Obtaining 1-out-of-k OT}

Our OT protocol can be easily extend to obtain 1-out-of-k OT, i.e., $\snd$ has $k$ input strings and $\rec$ can learn one of the strings that he chooses. The idea is that there are $k$ public-keys $(\pk_0,\ldots, \pk_{k-1})$ that are used to encrypted the strings using the same technique as before. In the modified protocol, $\rec$ generates a pair of keys $(\pk_c,\sk_c) \getsr \kg(1^{\secpar})$ for his choice value $c \in \{0,\ldots,k-1\}$. He also samples a random string $\seed \getsr \{0,1\}^{\secpar}$ and sends the queries $\msg{}{\seed\|1},\ldots,\msg{}{\seed\|k-1}$
to $\Frok$, obtaining $\msg{}{q_1},\ldots, \msg{}{q_{k-1}}$ as answers. $\rec$ then computes all $\pk_{i}$ for $i \in \{0,\ldots,k-1\} \setminus c$ in such way that $\pk_0 \star \pk_{j} =q_j$ for 
$j \in \{1,\ldots,k-1\}$. He sends $(\seed,\pk_0)$ to $\snd$. She uses the answers to the same random oracle queries as well as $\pk_0$ to reconstruct $(\pk_0,\ldots, \pk_{k-1})$.

\section{Adaptive Security}\label{sec:adp}
In this section we show that Protocol~$\pot$ is secure against adaptive adversaries. In the case of adaptive corruptions, the simulator has to handle adversaries that corrupt parties after the protocol execution has started, potentially after the execution is finished. When a party is corrupted in the ideal world, the simulator learns its inputs (and possibly outputs) and needs to hand it to its internal copy of the adversary along with the internal state of the dummy party (run internally by the simulator) corresponding to the corrupted party. This internal state must be consistent with both the inputs learned upon corruption and the messages already sent between dummy parties in the simulation. Usually it is hard to construct a simulator capable of doing so because it must first simulate a protocol execution with its internal copy of the adversary without knowing the inputs of uncorrupted parties and later, if a corruption happens, it must generate an internal state for the corrupted dummy party that is consistent with the newly learned inputs and the protocol messages that have already been generated. Intuitively, this requirement means that the simulator must simulate ``non-committing'' messages for honest dummy parties such that, upon corruption, it can generate randomness that would lead an honest party executing the protocol to generate the messages sent up to that point had it been given the inputs obtained form the ideal world party. 

In order to prove that Protocol~$\pot$ is indeed secure against adaptive adversaries, we will construct a simulator that generates messages in its internal execution with a copy of the adversary such that it can later come up with randomness that would result in these messages being generated given any input to any of the parties.
The general structure of this simulator is very similar to the simulator for the static case, using the same techniques for extracting inputs. 
The main modification in the simulator lies in the case where both the $\snd$ and $\rec$ are honest. In the proof of security against adaptive adversaries, the simulator no longer handles this case by simulating an interaction between two internal dummy parties running the protocol on random inputs. Instead, the simulator generates the first message (from $\rec$ to $\snd$) by acting as in the static case of a corrupted sender and the second message (from $\snd$ to $\rec$) by acting as in the static case of a corrupted receiver. From the proof of static security, this obviously generates a view that is indistinguishable from a real execution of the protocol. However, now the first message $(\seed,\pk_0)$ is such that both $\pk_0$ and $\pk_1$ (as defined by the protocol) are valid public-keys. Similarly, the second message $(\otp_0,\otp_1,\ciph_0,\ciph_1)$ contains random $\otp_0,\otp_1$ and $\ciph_0,\ciph_1$ containing $\padm_0, \padm_1$ that haven't yet been queried to $\Frop$.

Having these messages generated in simulating an execution between honest $\snd$ and $\rec$ will allow the simulator to ``explain'' the randomness used in an execution in case of an adaptive corruption.
In the case of an adaptive corruption of $\rec$, the simulator can hand the state of $\rec$ upon corruption to its internally executed adversary by simply giving the adversary $c$,$(\pk_c,\sk_c)$,$\seed$,$\msg{}{q}$,
$\pk_{1-c}$ and the random coins of $\kg(1^{\secpar})$ when generating $(\pk_c,\sk_c)$, given $c$ revealed when the the ideal world $\rec$ is corrupted. Since both $\pk_0$ and $\pk_1$ are valid and a valid public-key is indistinguishable from an invalid one, the state is consistent with $c$. In the case of an adaptive corruption of $\snd$, the simulator obtains $(\otmsg_0,\otmsg_1)$ by corrupting the ideal world $\snd$, sets the answers $\msg{}{\padm_0^{\prime}}$ and $\msg{}{\padm_1^{\prime}}$ to queries $\msg{}{\padm_0}$ and $\msg{}{\padm_1}$ to the simulated $\Frop$ such that $\otp_0=\padm_0^{\prime} \oplus \otmsg_0, \otp_1=\padm_1^{\prime} \oplus \otmsg_1$, and finally hands $\snd$'s state to the adversary as $\otmsg_0,\otmsg_1$, $\pk_0$, $\pk_1$ (computed as in the protocol), $\padm_0, \padm_1$, $\padm_0^{\prime}$, $\padm_1^{\prime}$ and message $(\otp_0,\otp_1,\ciph_0,\ciph_1)$.

\begin{theorem}\label{th:apotucme}
Let $\pke$ be a public-key encryption scheme that satisfies Properties \ref{ass:dowa} and \ref{ass:keys}. When instantiated with $\pke$, Protocol~$\pot$ UC-realizes the functionality $\Fot$ against adaptive adversaries in the $\Fro$-hybrid model.
\end{theorem}

\begin{proof}
As in the case of an static adversary, we must construct a simulator $\suc$ such that no environment $\env$ can distinguish between interactions with an adversary $\adv$ in the real world and with $\suc$ in the ideal world. 
For the sake of clarity, we will describe the simulator $\suc$ separately for the different corruption scenarios. In all cases, $\suc$ writes all the messages received from $\env$ in $\adv$'s input tape, simulating $\adv$'s environment. Also, $\suc$ writes all messages from $\adv$'s output tape to its own output tape, forwarding them to $\env$.

First we handle the four cases corresponding to the corruption statuses (\textit{i.e.} corrupted or honest) of $\snd$ and $\rec$ before the execution starts:\newline

\noindent\textbf{Case 1: Both $\snd$ and $\rec$ are honest:} This case is the main point where the simulator for an adaptive adversary changes in relation to the static adversary case. In this case, the simulator will simulate an execution between honest $\snd$ and $\rec$ using random inputs. The easiest way to perform this simulation is to run the protocol exactly as honest parties would when given random inputs, which is the done in the proof of Theorem~\ref{th:potucme}. However, for handling adaptive corruptions, the simulator will generate ``non-committing'' messages for honest dummy parties in such a way that the simulated execution is indistinguishable from a real execution and that the simulator can later generate consistent randomness to claim that honest parties had been executing the protocol with an arbitrary input. The simulator for this case is described in Figure~\ref{fig:suchonest}. \newline

\begin{boxfig}{Simulator $\suc$ for the case where both $\snd$ and $\rec$ are honest.}{fig:suchonest}
\begin{center}
\textbf{Simulator } $\suc$ (Honest $\snd$ and $\rec$)
\end{center}

Let $\lenf$ be the length of the messages and $\secpar$ be the security parameter.
The simulator $\suc$ interacts with an environment $\env$, functionality $\Fot$ and an internal copy $\adv$ of the adversary that does not corrupt $\snd$ or $\rec$, only observing their interaction and querying $\Frok$ and $\Frop$. $\suc$ proceeds as follows:

\begin{enumerate}
\item $\suc$ simulates the answers to the random oracle queries from $\adv$ exactly as $\Frok$ and $\Frop$ would (and stores the lists of queries/answers), except when stated otherwise in $\suc$'s description.

\item Simulating the first message (from dummy $\rec$):
    \begin{enumerate}
        \item $\suc$ generates two secret and public key pairs $(\pk_0,\sk_0) \getsr \kg(1^{\secpar})$ and $(\pk_1,\sk_1) \getsr \kg(1^{\secpar})$. 
$\suc$ samples $\seed \getsr \{0,1\}^{\secpar}$ and aborts if $\seed$ has been already queried to $\Frok$. $\suc$ sends $(\seed,\pk_0)$ to the dummy $\snd$.

        \item When $\adv$ queries $\Frok$ with $\msg{}{\seed}$, $\suc$ answers with $\msg{}{\pk_0 \oplus \pk_1}$.
    \end{enumerate}

\item Simulating the second message (from dummy $\snd$):

    \begin{enumerate}
        \item Samples $\otp_0,\otp_1 \getsr \{0,1\}^{\lenf}$, $\padm_0, \padm_1 \getsr \{0,1\}^{\secpar}$ and $\rnd_0,\rnd_1 \getsr \rndsp$. If $\padm_0$ or $\padm_1$ has already been queried to $\Frop$, $\suc$ aborts. Otherwise, $\suc$ computes $\ciph_0 \gets \enc(\pk_0,\padm_0,\rnd_0)$ and $\ciph_1 \gets \enc(\pk_1,\padm_1,\rnd_1)$, and sends $(\otp_0,\otp_1,\ciph_0,\ciph_1)$ to dummy $\rec$.
    \end{enumerate}
    
\item $\suc$ resumes answering the random oracle queries from $\adv$ as above, except for the queries $\padm_0$ and $\padm_1$ to $\Frop$. Upon receiving such query $\padm_c$, $\suc$ aborts. When $\adv$ halts, $\suc$ also halts and outputs whatever $\adv$ outputs.

\end{enumerate}

\end{boxfig}

\noindent\textbf{Case 2: $\snd$ is corrupted and $\rec$ is honest:} This case is equivalent to an execution with a static adversary that corrupts only $\snd$. Hence, the simulator proceeds as in the proof of Theorem~\ref{th:potucme}.\newline

\noindent\textbf{Case 3: $\snd$ is honest and $\rec$ is corrupted:} This case is equivalent to an execution with a static adversary that corrupts only $\rec$. Hence, the simulator proceeds as in the proof of Theorem~\ref{th:potucme}.\newline

\noindent\textbf{Case 4: Both $\snd$ and $\rec$ are corrupted:} This case is equivalent to an execution with a static adversary that corrupts both $\snd$ and $\rec$. Hence, the simulator proceeds as in the proof of Theorem~\ref{th:potucme}.\newline

Now we analyze the adaptive corruption cases. 
Notice that the simulator proceeds to simulate as in the static case after handing the internal state of the corrupted dummy party to the adversary. 
This is the case because the messages generated by the simulator for honest dummy parties already allow the simulator to continue simulating using the same instructions.
In the case of a corrupted $\snd$, the message $(\seed,\pk_0)$ generated for the dummy receiver is already such that both $\pk_0$ and $\pk_1$ (if computed as in the protocol) are valid public-keys for which the simulator knows the corresponding secret-keys, allowing it to extract the adversary's messages.
In the case of a corrupted $\rec$, the message $(\otp_0,\otp_1,\ciph_0,\ciph_1)$ generated for the dummy sender is already such that $\padm_0$ or $\padm_1$ have not been queried to $\Frop$, allowing the simulator to set the answer to these queries in such a way that $\otmsg_c = \otp_c \oplus \padm_c^{\prime}$ for any arbitrary $\otmsg_c \in \{0,1\}^{\lenf}$ and any $c \in \{0,1\}$.
Hence, we focus describing how the simulator generates the internal states of dummy $\snd$ and dummy $\rec$ in case the adversary chooses to corrupt them during the simulation. We will handle the cases of adaptive corruptions of $\snd$ and $\rec$ separately.\newline

\noindent\textbf{Adaptive corruption of $\snd$:} When $\snd$ is corrupted in the ideal world, $\suc$ obtains its input $(\otmsg_0,\otmsg_1)$. In the simulated execution with its internal copy of the adversary $\adv$, $\suc$ must return both the input and a consistent internal state for the dummy $\snd$. We further divide this scenario in two cases regarding the point of the execution when $\snd$ becomes corrupted:

\begin{itemize}
    \item Before the second message (sent by $\snd$): In this case, dummy $\snd$'s internal state consists solely of the inputs $(\otmsg_0,\otmsg_1)$. $\suc$ returns $(\otmsg_0,\otmsg_1)$ to $\adv$. 
    
    \item After the second message  (sent by $\snd$): In this case, $\suc$ must return to $\adv$ both the input $(\otmsg_0,\otmsg_1)$ and the internal state of dummy $\snd$ in a way that it is consistent with both the input and the message $(\otp_0,\otp_1,\ciph_0,\ciph_1)$ sent by dummy $\snd$. In both the cases of an honest $\rec$ and of a dishonest $\rec$, the second message is computed by $\suc$ in the same way and it knows $\padm_0, \padm_1$. Hence, $\suc$ computes $\padm_0^{\prime}=\otp_0 \oplus \otmsg_{0}$, $\padm_1^{\prime}=\otp_1 \oplus \otmsg_{1}$ and returns $(\otmsg_0,\otmsg_1,\pk_1,\padm_0,\padm_1,\padm_0^{\prime},\padm_1^{\prime},\rnd_0,\rnd_1)$ to $\adv$ as the dummy $\snd$'s state. Furthermore, if $\adv$ queries $\Frop$ with $\padm_i$, for $i \in \{0,1\}$, answers with the value $\otp_i \oplus \otmsg_{i}$.
\end{itemize}

If $\rec$ is honest, $\suc$ continues the simulation from the point where the corruption happened as in Case 2 (where only $\snd$ is corrupted). Otherwise, it continues as in Case 4 (where both $\snd$ and $\rec$ are corrupted).\newline

\noindent\textbf{Adaptive corruption of $\rec$:} When $\rec$ is corrupted in the ideal world, $\suc$ obtains its input $c$. In the simulated execution with its internal copy of the adversary $\adv$, $\suc$ must return both the input and a consistent internal state for the dummy $\rec$. We further divide this scenario in two cases regarding the point of the execution when $\rec$ becomes corrupted:

\begin{itemize}
    \item Before the first message (sent by $\rec$): In this case, dummy $\rec$'s internal state consists solely of the input $c$. $\suc$ returns $c$ to $\adv$. 
    
    \item Between the first message (sent by $\rec$) and the second message (sent by $\snd$): In this case, $\suc$ must return to $\adv$ both the input $c$ and the internal state of dummy $\rec$ in such a way that it is consistent with both the input and the first message $(\seed,\pk_0)$ sent by dummy $\rec$. In both the cases of an honest $\snd$ and of a dishonest $\snd$, the first message is computed by $\suc$ in the same way and he knows both key pairs $(\pk_0,\sk_0)$ and $(\pk_1,\sk_1)$, as well as the randomness used to generate each of them. Hence, $\suc$ computes $q=\pk_0 \oplus \pk_1$ and returns $(c,\sk_c,\pk_c,\rnd_c,q,\pk_{1-c})$ to $\adv$ as the internal state of dummy $\rec$, where $\rnd_c$ is the randomness used by $(\pk_c,\sk_c) \getsr \kg(1^{\secpar})$. Notice that in the case where $\snd$ and $\rec$ are honest as well as in the case where only $\snd$ is corrupted, $\suc$ would answer a query $\msg{}{\seed}$ to $\Frok$ with $\msg{}{q}$, where $q=\pk_0 \oplus \pk_1$. Hence, when $\adv$ queries $\Frok$ with $\msg{}{\seed}$, $\suc$ consistently answers with $\msg{}{q}$.
    
    \item After the second message (sent by $\snd$): In this case, $\suc$ also obtains $\rec$'s output $\otmsg_c$ in the ideal world. Hence, $\suc$ must return to $\adv$ the input $c$, randomness that is consistent with the first message and, additionally, randomness that is consistent with the second message. The first message does not depend on the output, so $\suc$ computes $(\sk_c,\pk_c,\rnd_c,q,\pk_{1-c})$ and handles queries to $\Frok$ as in the previous case.
    If $\snd$ is corrupted, $\suc$ retrieves $(\padm_c,\padm_c^{\prime})$ from the internal list of the simulated $\Frop$ (since either $\adv$ or $\suc$ made this query). Otherwise, if $\snd$ is honest, $\suc$ sets $\padm_c^{\prime}=\otp_c \oplus \otmsg_{c}$. Finally, $\suc$ returns $(\sk_c,\pk_c,\rnd_c,q,\pk_{1-c},\padm_c^{\prime})$ to $\adv$ as dummy $\rec$'s internal state. Furthermore, upon receiving a query $\padm_c$ from $\adv$ to $\Frop$, $\suc$ answers with the value $\otp_c \oplus \otmsg_{c}$.
\end{itemize}

If $\snd$ is honest, $\suc$ continues the simulation from the point where the corruption happened as in Case 3 (where only $\rec$ is corrupted). Otherwise, it continues as in Case 4 (where both $\snd$ and $\rec$ are corrupted).
\end{proof}

\section{Instantiations of the OT Protocol}\label{sec:inst}

\subsection{Instantiation Based on the McEliece Cryptosystem}\label{sec:mepke}

Let $n$, $k$ and $t$ be functions of the security parameter $\secpar$ (for easy of readability this will not be explicit in the notation). Consider a linear-error correcting code $C$ with length $n$ and dimension $k$, which consists of a $k$-dimensional subspace of $\Ft^{n}$, and 
let $G \in \Ft^{k \times n}$ denote the generator matrix of $C$. For the parameters $n$, $k$ and $t$, the assumption on the hardness of the bounded decoding problem can be stated as follows. 

\begin{assumption}\label{ass:bdec}
For parameters $n$, $k$ and $t$ that are functions of the security parameter $\secpar$, sample the generator matrix $G \getsr \Ft^{k \times n}$, the message $x \getsr \Ft^{k}$ and a uniformly random error $e \in \Ft^{n}$ such that $\HW(e)= t$. Then for $y=xG\xor e$ and for every PPT decoder $\adv$ 
\[
\Pr[\adv(G,y) = x ]  \in \mathsf{negl}(\secpar).
\]
\end{assumption}

We next prove some useful facts that will be used in our constructions, these observations follow the lines of Mathew et al. \cite{MVVR12}. Lets call a fixed generator matrix $G$ good if there exists a PPT decoder that can recover the message $x$ with non-negligible probability when encoded with this specific $G$. Clearly if the assumption on the hardness of the bounded decoding problem holds, then the subset of good generator matrices constitutes only a negligible fraction of all generator matrices.

\begin{observation}[\cite{MVVR12}]\label{obs:den}
The fraction of matrices $Q \in \Ft^{k \times n}$ that can be expressed as $Q=G_1\xor G_2$ for good generator matrices $G_1$ and $G_2$ is negligible.
\end{observation}

\begin{proof}
This follows from a simple counting argument: if there are $g$ good generator matrices, there can be at most $\binom{g}{2}=O(g^2)$ matrices $Q$ that can be expressed as the exclusive-or of two good generator matrices.
\end{proof}

One possible instantiation of our OT protocol uses the McEliece cryptosystem, which has $\pksp = \Ft^{k \times n}$, $\messp = \Ft^k$, $\ciphsp=\Ft^n$, and works as follows:

\begin{itemize}
\item Key generation: Generate a generator matrix $G \in \Ft^{k \times n}$ of a Goppa code together with the efficient error-correction algorithm $\mathsf{Correct}$ that can correct up to $t$ errors. Generate a uniformly random non-singular matrix $S \in \Ft^{k \times k}$ and a uniformly random permutation matrix $T \in \Ft^{n \times n}$. Set $\pk=SGT$ and $\sk=(S,G,T)$.

\item Encryption: Given as input the public-key $\pk$ and the message $\mes \in \Ft^k$, sample uniformly at random an error vector $e \in \Ft^{n}$ such that $\HW(e)= t$ and output the ciphertext $\ciph \gets \mes \cdot \pk \xor e$.

\item Decryption: Given the ciphertext $\ciph$ and the secret-key $\sk$ as input, compute $\ciph \cdot T^{-1}=(\mes S)G\xor eT^{-1}$. Then compute 
$\mes S \gets \mathsf{Correct}(cT^{-1})$ and output $\mes = (\mes S)S^{-1}$.
\end{itemize}

The security of the McEliece cryptosystem relies on Assumption \ref{ass:bdec}. Therefore, when considering the operation $\xor$ on $\pksp$, Observation \ref{obs:den} implies that Property \ref{ass:dowa} holds straightforwardly if 
Assumption \ref{ass:bdec} is true. The security of the McEliece encryption scheme is also based on the following assumption regarding the pseudorandomness of the public-keys, which is equivalent to Property \ref{ass:keys} for this cryptosystem.

\begin{assumption}\label{ass:goppa}
Let $(\pk,\sk)\getsr \kg(1^\secpar)$ be the key generation algorithm of the McEliece cryptosystem for $\pk \in \Ft^{k \times n}$. Let $R \getsr \Ft^{k \times n}$. For every PPT distinguisher $\adv$ that outputs a decision bit it holds that 
\[
| \Pr[\adv(\pk) = 1 ] - \Pr[\adv(R) = 1 ] | \in \mathsf{negl}(\secpar).
\]
\end{assumption}

\subsection{QC-MDPC based Instantiation}
The QC-MDPC cryptosystem \cite{misoczki2013mdpc} is a variant of the McEliece cryptosystem that has much shorter keys. Currently, it is the most efficient code-based public-key cryptosystem. We briefly describe it here. In the following $\wt(x)$ denotes the Hamming weight of a vector $x$. 


\begin{itemize}
\item Key generation: Let $\cir(v)$ denote the circulant matrix whose first row is $v \in \F_2^r$, a binary array of length $r$. A QC-MDPC secret-key is a sparse parity-check matrix of form $\hat{H} = \left[ \cir(f) \mid \cir(g) \right]$ where $\gcd(g, x^r - 1) = 1$ and $\wt(f) + \wt(g) = t$ (for some suitable choice of $r$ and $t$), and the corresponding public-key is the systematic parity-check matrix $H = \left[ \cir(h) \mid I \right]$ where $h = f \cdot g^{-1} \mod x^r - 1$, or equivalently the systematic generator $G = \left[ I \mid \cir(h)^T \right]$, both represented by $h \in \F_2^r$. .

\item Encryption: We encrypt a message by encoding it as a vector $e = (e_0, e_1) \in \F_2^{2r}$ of weight $\wt(e) = t$, choosing a uniformly random $m \samples \F_2^r$ and computing the ciphertext as $c \gets mG + e \in \F_2^{2r}$.

\item Decryption: Let $\Psi_{\hat{H}}$ denote a $t$-error correcting decoding algorithm that has knowledge of the 
secret-key $\hat{H}$. 
Extract the error vector $e$ by decoding $c \gets \Psi_{\hat{H}}(c)$.

\end{itemize}

\begin{assumption}
The \emph{decisional QC-MDPC assumption} states that distinguishing $h$ from a uniformly random vector $u \in \F_2^r$ is unfeasible.
\end{assumption}

The OW-CPA security of this cryptosystem follows from the pseudorandomess of the public-key (the decisional QC-MDPC assumption) and the hardness of the bounded decoding problem (Assumption~\ref{ass:bdec}), exactly as in the original McEliece. 

The QC-MDPC cryptosystem as described here does satisfy the requirements to be used within our general construction. The analysis is exactly the same as in the original McEliece PKC in Section \ref{sec:mepke} and boils down to two arguments: 
\begin{itemize}
    \item The fraction of vectors $q \in \F_2^r$ that can be expressed as $q=h_1\xor h_2$ for good vectors $h_1$ and $h_2$ is negligible. Here, a good vector $h_i, i \in \{0,1\}$ is one where the corresponding parity-check matrix $H_i = \left[ \cir(h) \mid I \right]$ has a trapdoor. This fact and the hardness of the bounded decoding problem (Assumption \ref{ass:bdec}) imply Property \ref{ass:dowa}. 
    \item The pseudorandomness of the public-key implies Property \ref{ass:keys}. 
\end{itemize}


\subsection{LPN-based Instantiation}\label{sec:lpnpke}

We also present a solution based on the Learning Parity with Noise (LPN) problem. This problem essentially states that given a system of binary linear equations in which the outputs are disturbed by some noise, it is difficult to determine the solution. In this work we use the low noise variant of LPN as first studied by Alekhnovich in~\cite{FOCS:Alekhnovich03}. The following equivalent statement of the assumption is from D\"{o}ttling et al.~\cite{AC:DotMulNas12}.

\begin{assumption}\label{ass:lpn}
Let the problem parameter be $n \in \mathbb{N}$, which is a function of the security parameter $\secpar$. Let $m=O(n)$, $\epsilon >0$ and
$\rho=\rho(n)=O(n^{-1/2-\epsilon})$. Sample $B \getsr \Ft^{m \times n}$, $x \getsr \Ft^{n}$ and $e \getsr \chi_\rho^m$. The problem is, given $B$ and $y \in \Ft^m$,
to decide whether $y=Bx+e$ or $y \getsr \Ft^m$. The assumption states that for every PPT distinguisher $\adv$ that outputs a decision bit it holds that 
\[
| \Pr[\adv(B,y=Bx+e) = 1 ] - \Pr[\adv(B,y \getsr \Ft^m) = 1 ] | \in \mathsf{negl}(\secpar).
\]
\end{assumption}

The current best distinguishers for this problem require time of the order $2^{\secpar^{1/2-\epsilon}}$ and 
for this reason by setting $n = O(\secpar^{2/(1-2\epsilon)})$ where $\secpar$ is the security parameter of the 
encryption scheme the hardness is normalized to $2^{\Theta(\secpar)}$.

Our OT protocol can be instantiated based on the IND-CPA secure cryptosystem of D{\"o}ttling et al.~\cite{AC:DotMulNas12} 
that is based on the low noise variant of LPN. We should emphasize that we just need the simplest 
version of their cryptosystem, which is described below. The IND-CPA security of their cryptosystem is based on Assumption \ref{ass:lpn}. 

Consider the security parameter $\secpar$, and let $n, \ell_1, \ell_2 \in O(\secpar^{2/(1-2\epsilon)})$ and  $\rho \in O(\secpar^{-(1+2\epsilon)/(1-2\epsilon)})$ so that Assumption \ref{ass:lpn} is believed to hold. Let 
$G \in \Ft^{\ell_2 \times n}$ 
be the generator-matrix of a binary linear error-correcting code $C$ and 
$\mathsf{Decode}_C$ an efficient decoding procedure for $C$ that
corrects up to $\alpha \ell_2$ errors for constant $\alpha$.

\begin{itemize}

\item Key Generation: Let $A \getsr \Ft^{\ell_1 \times n}$, $T \getsr \chi_\rho^{\ell_2 \times \ell_1}$ and $X \getsr \chi_\rho^{\ell_2 \times n}$.
Set $B=TA+X$, $\pk=(A, B)$ and $\sk=T$. Output $(\pk, \sk)$.

\item Encryption: Given a message $\mes \in \Ft^{n}$ and the public-key $\pk=(A, B)$ as input, sample $s \getsr \chi_\rho^{n}$, $e_1 \getsr \chi_\rho^{\ell_1}$ and $e_2 \getsr \chi_\rho^{\ell_2}$. Set $\ciph_1=As+e_1$ and $\ciph_2=Bs+e_2+G\mes$. Output $\ciph=(\ciph_1,\ciph_2)$.

\item Decryption: Given a ciphertext $\ciph=(\ciph_1,\ciph_2)$ and a secret-key $\sk=T$ as input,
compute $y \gets \ciph_2 - T\ciph_1$ and $\mes \gets \mathsf{Decode}_C(y)$. Output $\mes$.
\end{itemize}

If Assumption \ref{ass:lpn} holds, then Property \ref{ass:keys} trivially holds for this cryptosystem as the public-keys are pseudorandom \cite{AC:DotMulNas12}. Based on Assumption \ref{ass:lpn}, D\"ottling et al. \cite{AC:DotMulNas12} also proved that this cryptosystem is IND-CPA secure \cite{AC:DotMulNas12} (which is a stronger notion than OW-CPA). By letting $\pksp=\Ft^{\ell_1 \times n} \times \Ft^{\ell_2 \times n}$, considering the operation $\xor$ on $\pksp$ and using a counting argument about the good public-keys as in Section \ref{sec:mepke}, we obtain that Property \ref{ass:dowa} trivially holds for this cryptosystem if Assumption \ref{ass:lpn} holds.

\subsection{Instantiation based on the CDH assumption}\label{sec:cdh}
We will instantiate our scheme by showing that the ElGamal cryptosystem is OW-CPA secure and has Properties \ref{ass:dowa} and \ref{ass:keys} under the Computational Diffie-Hellman (CDH) assumption. First we will recall the CDH assumption:

\begin{assumption}\label{ass:cdh}
The Computational Diffie-Hellman assumption requires that for every PPT adversary $\adv$
it holds that 
\[
\Pr[\adv(\G,w,g,g^a,g^b) = g^{ab} ] \in \mathsf{negl}(\secpar).
\]
where the probability is taken over the experiment of generating a group $\G$ of order $w$ with a generator $g$ on input $1^\secpar$ and choosing $a,b \getsr \Z_q$.
\end{assumption}

The classical ElGamal cryptosystem~\cite{C:ElGamal84} is parametrized by a group $(\G,g,w)$ of order $w$ with generator $g$ where the CDH assumption holds. We assume that $(\G,g,w)$ is known by all parties. The cryptosystem consists of a triple of algorithms $\pke=(\kg,\enc,\dec)$ that proceed as follows:

\begin{itemize}

    \item $\kg$ samples $\sk \getsr \Z_w$, computes $\pk=g^{\sk}$ and outputs a secret and public-key pair $(\pk,\sk)$.
    
    \item $\enc$ takes as input a public-key $\pk$ and a message $\mes \in \G$, samples $\rnd \getsr \Z_p$ computes $c_1=g^\rnd$, $c_2=m \cdot \pk^{\rnd}$ and outputs a ciphertext $\ciph=(c_1,c_2)$.
    
    \item $\dec$ takes as input a secret-key $\sk$, a ciphertext $\ciph$ and outputs a message $\mes = c_2/c_1^{\sk}$.
\end{itemize}

The ElGamal cryptosystem described above is well-known to be OW-CPA secure~\cite{C:ElGamal84}, leaving us to prove that it has Properties~\ref{ass:dowa} and~\ref{ass:keys}.
Property~\ref{ass:keys} follows trivially from the fact that $\pk$ is chosen uniformly over all elements of $\G$.

\begin{observation}\label{obs:cdhdowa}
The ElGamal cryptosystem described above has Property~\ref{ass:dowa} under the CDH assumption.
\end{observation}

\begin{proof}
First we observe that $\pksp$ is $\G$, which is a group. Assume by contradiction that an adversary $\adv$ succeeds in the experiment of Property~\ref{ass:dowa}. Under the CDH assumption, $\adv$ must know both $\sk_1$ and $\sk_2$ corresponding to $\pk_1$ and $\pk_2$. However, we know that $\pk_1 \cdot \pk_2=q$ for a uniformly random $q \getsr \G$ (using multiplicative notation for $\G$). If $\adv$ freely generated $\pk_1$ and $\pk_2$ such that $\pk_1 \cdot \pk_2=q$ and knows secret-keys $\sk_1$ and $\sk_2$, then it knows the discrete logarithm of $q$, since it is equal to $\sk_1+\sk_2$. The CDH assumption implies that computing discrete logarithms is hard, hence we have a contradiction and the observation holds.
\end{proof}

\subsubsection{A CDH based instantiation matching the efficiency of Simplest OT \texorpdfstring{~\cite{LC:ChoOrl15}}{(Chou and Orlandi Latincrypt 2015)}}
If we instantiate Protocol~\pot with the ElGamal cryptosystem described above in a black-box way, the sender $\snd$ has to compute and send both ciphertexts $\ciph_0=(g^{\rnd_0},\otmsg_0 \cdot \pk_0^{\rnd_0})$ and $\ciph_1=(g^{\rnd_1},\otmsg_1 \cdot \pk_1^{\rnd_1})$, which amounts to 4 exponentiations and 4 group elements. However, notice that $\otmsg_0$ and $\otmsg_1$ are being encrypted under two different public-keys $\pk_0$ and $\pk_1$. Hence, we can invoke a result by Bellare \textit{et al.}~\cite{PKC:BelBolSta03} showing that ciphertexts $\ciph_0$ and $\ciph_1$ can be computed with the same randomness $\rnd$. This simple observation can increase both computational and communication efficiency, since terms of the ciphertexts that depend only on the randomness only need to be computed and sent once.

Basically, the result of Bellare \textit{et al.}~\cite{PKC:BelBolSta03} shows that randomness can be reused while maintaining the same underlying security guarantees when encrypting multiple messages under multiple different public-keys with ``reproducible" cryptosystems, of which the ElGamal cryptosystem is an example as proven in~\cite{PKC:BelBolSta03}. Bellare \textit{et al.} provide a modular generic reduction showing that it can be used to prove that reproducible cryptosystems remain IND-CPA or IND-CCA secure under randomness reuse. Since the ElGamal cryposystem is IND-CPA secure under the DDH assumption, we can directly optimize the ElGamal based instantiation above under the DDH assumption. However, we remark that  the reduction of Bellare \textit{et al.} is generic (depending only on the cryptosystem being reproducible) and straightforward to adapt to OW-CPA security, which is naturally implied by IND-CPA security. Hence, applying the techniques of Bellare \textit{et al.} we can show that the ElGamal cryptosystem remains OW-CPA secure with randomness re-use under the CDH assumption when multiple different messages are encrypted under multiple different keys because it is a reproducible cryptosystem (as already proven in~\cite{PKC:BelBolSta03}). Due to space restrictions and the lack of novelty, we leave a full discussion of the application of the result of Bellare \textit{et al.}~\cite{PKC:BelBolSta03} to OW-CPA secure cryptosystems to the full version of this paper.

Reusing the same randomness for both $\ciph_0$ and $\ciph_1$ means that now $\snd$ only has to compute $g^\rnd$ once and send $g^\rnd,\otmsg_0 \cdot \pk_0^{\rnd},\otmsg_1 \cdot \pk_1^{\rnd}$, meaning that it saves 1 exponentiation in terms of computation and 1 group element in terms of communication. In fact, the total number of exponentiations in this instantiation is only $5$ matching the computational efficiency of the Simplest OT protocol of Chou and Orlandi~\cite{LC:ChoOrl15}, which is also proven secure in the ROM but only provides security against static adversaries DDH assumption over a gap-DH group. On the other hand, this instantiation requires two extra group elements to be communicated. Nevertheless, at this small cost in communication, this instantiation provides security against adaptive adversaries under a weaker assumption.

\subsection{LWE-based Instantiation}\label{sec:lwe}

In order to instantiate our scheme under the LWE assumption, we observe that an IND-CPA secure (and thus OW-CPA secure) public-key encryption scheme proposed by Peikert \textit{et al.}~\cite[Section 7.2]{C:PeiVaiWat08} satisfies Properties~\ref{ass:dowa} and~\ref{ass:keys}. First, notice that it is proven to have Property~\ref{ass:keys} in~\cite[Lemma 7.3]{C:PeiVaiWat08}. The public-key space is $\pksp = \Z_q^{m \times n} \times \Z_q^{m \times l}$, forming a group with its addition operation $+$. It is pointed out \cite[Section 7.2]{C:PeiVaiWat08} that the public-keys of this cryptosystem are such that they embed several samples of the LWE distribution. Informally, an adversary who succeeds in the experiment of Property \ref{ass:dowa} must know both secret-keys $\sk_1$ adn $\sk_2$ associated to $\pk_1$ and $\pk_2$, respectively. If we have that $\pk_1 + \pk_2 = q$ for a random $q \getsr \pksp$, an adversary who manages to learn the corresponding $\sk_1$ and $\sk_2$ is able to decide whether $q$ is an instance of the LWE distribution or not, which contradicts the LWE assumption. Due to the lack of space, we leave a formal argument to the full version of this paper.

\subsection{Other Instantiations} 

In the full version of the paper we will present instantiations based on Quadratic Residuosity (QR), Decisional Composite Residuosity (DCR) and NTRU assumptions.

\section{Implementation}
Here, we propose concrete parameters and present implementation results for the QC-MDPC instantiation of our OT protocol. As the total cost of the entire OT protocol will be dominated by the key generation, encryption and decryption costs of the underlying public-key cryptosystem, we restrict our presentation to these respective costs. 

\subsection{Parameters}
We computed the classical security level of a few parameter sets for QC-MDPC McEliece using a SAGE script. The quantum level is one half of the values here presented (e.g. to obtain 128-bit quantum secure parameters, take the 256-bit classically secure ones). The  communication complexity is $2n$ bits for each party. Results are shown in table \ref{table:1}.

\begin{table}[h!]
\centering
 \begin{tabular}{|c | c | c | c | c|} 
 \hline
  Security Level & Codimension $r$ & Code length $n$ &Row density $w$ & Introduced errors $t$  \\  
 \hline
 256 & 32771 & 2*32771 & 2*137 & 264 \\ 
 192 & 19853 & 2*19853 & 2*103 & 199 \\
 128 & 10163 & 2*10163 & 2*71 & 134  \\ 
 \hline
 \end{tabular}
\label{table:1}
\vspace{0.08in}
\caption{Parameters for QC-MDPC PKC}
\end{table}
\vspace{-0.2in}
We recall that it is necessary that $r$ be prime and that the polynomial $(x^r - 1)/(x - 1)$ be irreducible over GF(2). There exists a trade-off between communication costs and decoding complexity. For instance, these parameters are quite tight and decoding may be somewhat slow. By increasing r to, say, r = 49139 (50\% more bandwidth), decoding speed roughly quadruples (encoding becomes about half as fast, but it is far faster than decoding to begin with).
\subsection{Implementation}

Some running times are shown in Table \ref{table:2}. These results are for a (non-vectorized) C implementation, compiled with 64-bit gcc under Windows 10, on an Intel i7-5500U @ 2.4 GHz with TurboBoost disabled. Running times are measured in kilocycles. The underlying PRNG is ChaCha20. The implementation is mostly isochronous, to prevent side-channel timing attacks.

\begin{table}[h!]
\centering
 \begin{tabular}{|c | c | c | c|} 
 \hline
  Security Level & Key generation & Encryption $n$ &Decryption  \\ 
 \hline
 256 & 507.48 kcyc & 483.51 kcyc & 10,337.41 kcyc \\ 
 128 & 101.59 kcyc & 73.85 kcyc & 1,758.66 kcyc \\
 80 & 67.45 kcyc & 52.49 kcyc & 497.53 kcyc  \\  
 \hline
 \end{tabular}
 \vspace{0.08in}
 \caption{Running Times (in kilocycles) for QC-MDPC PKC.}
\label{table:2}
\end{table}

\section{Conclusions}\label{sec:con}

In this work we presented a framework for obtaining efficient round-optimal UC-secure OT protocols that are secure against active adaptive adversaries. Our construction can be instantiated with the low noise LPN, McEliece,  QC-MDPC, LWE, and CDH assumptions. Our instantiations based on the low noise LPN,  McEliece, and QC-MDPC assumptions (which are more efficient than the previous works) are the first OT protocols, which are UC-secure against active adaptive adversaries based on these assumptions. 
Our CDH-based instantiation has basically the same efficiency as the Chou-Orlandi Simplest OT protocol and only two extra group elements of communication, but achieves security against stronger adversaries under a weaker assumption.
In the full version, we will show that it is also possible to get instantiations based on QR, DCR and NTRU assumptions.

\end{document}